\documentclass[11pt,a4paper]{article}

\usepackage{amsmath,amsfonts,amssymb,amsthm}
\usepackage{amsthm}
\usepackage{graphicx,color}
\usepackage{boxedminipage}
\usepackage[ruled,boxed,linesnumbered]{algorithm2e}
\usepackage{framed}
\usepackage{thmtools}
\usepackage{thm-restate}
\usepackage{xspace}
\usepackage{todonotes}
\usepackage{footnote}

\usepackage{vmargin}
\setmarginsrb{1in}{1in}{1in}{1in}{0pt}{0pt}{0pt}{6mm}
\usepackage{todonotes}
\usepackage{footnote}
 \usepackage[pdftex, plainpages = false, pdfpagelabels, 
                 bookmarks=false,
                 bookmarksopen = true,
                 bookmarksnumbered = true,
                 breaklinks = true,
                 linktocpage,
                 pagebackref,
                 colorlinks = true,  
                 linkcolor = blue,
                 urlcolor  = blue,
                 citecolor = red,
                 anchorcolor = green,
                 hyperindex = true,
                 hyperfigures
                 ]{hyperref} 
 \usepackage{xifthen}
 \usepackage{tabularx}
\usepackage{tikz} 
 \usetikzlibrary{calc}
 
  \usepackage[nameinlink]{cleveref}

\newtheorem{theorem}{Theorem}

\newtheorem{definition}{Definition}

\newtheorem{proposition}{Proposition}
\newtheorem{redrule}{Reduction Rule}

\newtheorem{clm}{Claim}[theorem]

\theoremstyle{definition}

\newcommand{\cl}{{\sf cl}}

\DeclareMathOperator{\operatorClassNP}{NP}
\newcommand{\classNP}{\ensuremath{\operatorClassNP}\xspace}

\DeclareMathOperator{\operatorClassCoNP}{coNP}
\newcommand{\classCoNP}{\ensuremath{\operatorClassCoNP}}
\DeclareMathOperator{\operatorClassFPT}{FPT\xspace}
\newcommand{\classFPT}{\ensuremath{\operatorClassFPT}\xspace}
\DeclareMathOperator{\operatorClassW}{W}
\newcommand{\classW}[1]{\ensuremath{\operatorClassW[#1]}}

\newlength{\RoundedBoxWidth}
\newsavebox{\GrayRoundedBox}
\newenvironment{GrayBox}[1]%
   {\setlength{\RoundedBoxWidth}{.93\textwidth}
    \def\boxheading{#1}
    \begin{lrbox}{\GrayRoundedBox}
       \begin{minipage}{\RoundedBoxWidth}}%
   {   \end{minipage}
    \end{lrbox}
    \begin{center}
    \begin{tikzpicture}%
       \node(Text)[draw=black!20,fill=white,rounded corners,%
             inner sep=2ex,text width=\RoundedBoxWidth]%
             {\usebox{\GrayRoundedBox}};
        \coordinate(x) at (current bounding box.north west);
        \node [draw=white,rectangle,inner sep=3pt,anchor=north west,fill=white] 
        at ($(x)+(6pt,.75em)$) {\boxheading};
    \end{tikzpicture}
    \end{center}}

\newenvironment{defproblemx}[2][]{\noindent\ignorespaces%
                                \FrameSep=6pt%
                                \parindent=0pt%
                \vspace*{-1.5em}
                \ifthenelse{\isempty{#1}}{%
                  \begin{GrayBox}{\textsc{#2}}%
                }{%
                  \begin{GrayBox}{\textsc{#2} parameterized by~{#1}}%
                }
                \begin{tabular*}{\textwidth}{@{\hspace{.1em}} >{\itshape} p{1.8cm} p{0.8\textwidth} @{}}%
            }{
                \end{tabular*}%
                \end{GrayBox}%
                \ignorespacesafterend
            }

\newcommand{\defproblema}[3]{
  \begin{defproblemx}{#1}
    Input:  & #2 \\
    Task: & #3
  \end{defproblemx}
}%

\newcommand{\Oh}{\mathcal{O}}

\newcommand{\cM}{\mathcal{M}}
\newcommand{\cI}{\mathcal{I}}

\newcommand{\pname}{\textsc}
\newcommand{\ProblemFormat}[1]{\pname{#1}}
\newcommand{\ProblemIndex}[1]{\index{problem!\ProblemFormat{#1}}}
\newcommand{\ProblemName}[1]{\ProblemFormat{#1}\ProblemIndex{#1}{}\xspace}

\newcommand{\probStab}{\ProblemName{Independent Stable Set}}
\newcommand{\probIS}{\ProblemName{Stable Set}}
\newcommand{\probMIS}{\ProblemName{Rainbow-Stable Set}}

\begin{document}

\title{Stability in Graphs with Matroid Constraints\thanks{The research leading to these results has received funding from the Research Council of Norway via the project  BWCA
(grant no. 314528) and the European Research Council (ERC) via grant LOPPRE, reference 819416.}
}

\author{
Fedor V. Fomin\thanks{
Department of Informatics, University of Bergen, Norway.}
\and
Petr A. Golovach\addtocounter{footnote}{-1}\footnotemark{}
\and
Tuukka Korhonen\addtocounter{footnote}{-1}\footnotemark{}
\and
Saket Saurabh\addtocounter{footnote}{-1}\footnotemark{} \thanks{The Institute of Mathematical Sciences, HBNI, Chennai, India}
}

\date{}

\maketitle

\begin{abstract}
We study the following \probStab problem. Let $G$ be an undirected graph and $\cM = (V(G), \mathcal{I})$ be a matroid whose elements are the vertices of $G$.
For an integer $k\geq 1$, the task is to decide whether $G$ contains a set $S\subseteq V(G)$ of size at least $k$ which is independent (stable) in $G$ and independent in $\cM$. This problem generalizes several well-studied algorithmic problems, including  \textsc{Rainbow Independent Set}, 
\textsc{Rainbow Matching}, and  \textsc{Bipartite Matching with Separation}.  
We show that 
\begin{itemize}
\item 
When the matroid $\cM$ is represented by the independence oracle, then for any computable function $f$, no algorithm can solve \probStab using $f(k) \cdot n^{o(k)}$ calls to the oracle.
 \item
On the other hand, when the graph $G$ is of degeneracy $d$, then the problem is solvable in time $\Oh((d+1)^k \cdot n)$, and hence is  \classFPT parameterized by $d+k$. Moreover, when the degeneracy $d$ is a constant (which is not a part of the input), the problem admits a kernel polynomial in $k$. 
 More precisely, we prove that  
for every integer $d\geq 0$,  the problem admits a kernelization algorithm that in time $n^{\Oh(d)}$  outputs an equivalent framework with a  graph on $dk^{\Oh(d)}$ vertices. A lower bound complements this when $d$ is part of the input:   \probStab does not admit a polynomial kernel when parameterized by $k+d$ unless $\classNP \subseteq \classCoNP/{\rm poly}$. This lower bound holds even when $\cM$ is a partition matroid. 
\item  Another set of results concerns the scenario when the graph $G$ is chordal. In this case, our computational lower bound excludes an \classFPT algorithm when the input matroid is given by its independence oracle. However,  we demonstrate that \probStab can be solved in $2^{\Oh(k)}\cdot \|\cM\|^{\Oh(1)}$ time when $\cM$ is a linear matroid given by its representation. In the same setting, \probStab does not have a polynomial kernel when parameterized by $k$ unless $\classNP\subseteq\classCoNP/{\rm poly}$.
\end{itemize}
\end{abstract}

\section{Introduction}\label{sec:intro}
We initiate the algorithmic study of computing stable (independent) sets in frameworks. The term \emph{framework}, also known as \emph{pregeometric graph} \cite{Lovasz77,Lovasz19},  refers to a pair $(G, \cM)$, where $G$ is a graph and $\cM = (V(G), \mathcal{I})$ is a matroid on the vertex set of $G$.
We remind the reader that pairwise nonadjacent vertices of a graph form a \emph{stable} or \emph{independent} set. To avoid confusion with independence in matroids, we consistently use the term "stable set" throughout the paper. Whenever we mention independence, it is in reference to independence with respect to a matroid. We consider the following problem.

 \defproblema{\probStab}%
{A framework $(G,\cM)$ and an integer $k\geq 0$.}%
{Decide whether $G$ has vertex set $S\subseteq V(G)$ of size at least $k$ that is stable in $G$ and independent in  $\cM$.
}

The \probStab problem encompasses several well-studied problems related to stable sets.

When $\cM$ is a uniform matroid with every $k$-element subset of $V(G)$ forming a basis, the \probStab problem seeks to determine whether a graph $G$ contains a stable set of size at least $k$. This is the classic  \probIS problem. 

For a partition matroid $\cM$ whose elements are partitioned into $k$ blocks and independent sets containing at most one element from each block, \probStab transforms into the rainbow-independence (or \probMIS) problem. To express this problem in graph terminology, consider a graph $G$ with a vertex set $V(G)$ colored in $k$ colors. A set of vertices $S$ is termed \emph{rainbow-independent} if it is stable in $G$ and no color occurs in $S$ more than once~\cite{aharoni2023rainbow, kim2022rainbow}. This concept is also known in the literature as an \emph{independent transversal} \cite{GrafH20, GrafHH22, haxell2011forming} and an independent system of representatives \cite{aharoni2007independent}. 

Rainbow-independence generalizes the well-studied combinatorial concept of rainbow matchings \cite{AharoniBCHS19, DRISKO1998181}. (Note that a matching in a graph is a stable set in its line graph.) It also has a long history of algorithmic studies. In the \textsc{Rainbow Matching} problem, we are given a graph $G$, whose edges are colored in $q$ colors, and a positive integer $k$. The task is to decide whether a matching of size at least $k$ exists whose edges are colored in distinct colors.
 Itai,  Rodeh,  and Tanimoto in \cite{ItaiRT78}  established that \textsc{Rainbow Matching} is \classNP-complete on bipartite graphs.  Le and Pfender \cite{LeP14} strongly enhanced this result by showing that  \textsc{Rainbow Matching}  is \classNP-complete even on paths and complete graphs.
Gupta et al. \cite{GuptaRSZ19} considered the parameterized complexity of  \textsc{Rainbow Matching}. 
They gave an \classFPT  algorithm of running time $2^k \cdot n^{\Oh(1)}$. They also provided a kernel with  $\Oh(k^2 \Delta)$ vertices, where $\Delta$ is the maximum degree of a graph. Later, in 
\cite{GuptaRSZ20}, the same set of authors obtained a kernel with   $\Oh(k^2)$ vertices
for  \textsc{Rainbow Matching} on general graphs. 

When $\cM$ is a transversal matroid, \probStab transforms into the \textsc{Bipartite Matching with Separation} problem~\cite{ManurangsiSS23}. In this variant of the maximum matching problem, the goal is to determine whether a bipartite graph $H$ contains a matching of size $k$ with a separation constraint: the vertices on one side lie on a path (or a grid), and two adjacent vertices on a path (or a grid) are not allowed to be matched simultaneously. This problem corresponds to \probStab on a framework $(G, \cM)$, where $G$ is a path (or a grid) on vertices $U$, and $\cM$ is a transversal matroid of the bipartite graph $H = (U, W, E_H)$ whose elements are $U$, and the independent subsets are sets of endpoints of matchings of $H$. Manurangsi,  Segal-Halevi, and  Suksompong in  \cite{ManurangsiSS23} 
proved that \textsc{Bipartite Matching with Separation} is \classNP-complete and provided approximation algorithms.

 \probIS is a notoriously difficult computational problem. It is well-known to be \classNP-complete and \classW{1}-complete when parameterized by $k$~\cite{CyganFKLMPPS15}. On the other hand,  \probIS is solvable in polynomial time on perfect graphs \cite{grotschel2012geometric}. When it comes to parameterized algorithms and kernelization,  \probIS is known to be \classFPT and to admit polynomial (in $k$) kernel on classes of sparse graphs, like graphs of bounded degree or degeneracy \cite{CyganGH17}. The natural question is which algorithmic results about the stable set problem could be extended to \probStab.

\begin{itemize}
\item 
We commence with a lower bound on \probStab.  \Cref{thm:lb-uncond} establishes that when the matroid in a framework is represented by the independence oracle, for any computable function $f$, no algorithm can solve \probStab using $f(k) \cdot n^{o(k)}$ calls to the oracle. Moreover, we show that the lower bound holds for frameworks with bipartite, chordal, claw-free graphs, and AT-free graphs for which the classical \probIS problem can be solved in polynomial time. 
While the usual bounds in parameterized complexity are based on the assumption $\classFPT\neq\classW{1}$,  \Cref{thm:lb-uncond} 
rules out the existence of an \classFPT algorithm for \probStab parameterized by $k$ unconditionally.

\item
In \Cref{sec:deg}, we delve into the parameterized complexity of \probStab when dealing with frameworks on $d$-degenerate graphs. The first result of this section, \Cref{thm:fpt-degen}, demonstrates that the problem is \classFPT when parameterized by $d+k$, by providing an algorithm of running time  $\Oh((d+1)^k \cdot n)$. Addressing the kernelization aspect,  \Cref{thm:kern-degen} reveals that when $d$ is a constant, \probStab on frameworks with graphs of degeneracy at most $d$, admits a kernel polynomial in $k$. 
More precisely, we prove that  
for every integer $d\geq 0$,  the problem admits a kernelization algorithm that in time $n^{\Oh(d)}$  outputs an equivalent framework with a  graph on $dk^{\Oh(d)}$ vertices. This is complemented by  \Cref{thm:nokern-degen}, establishing that \probStab on frameworks with $d$-degenerate graphs and partition matroids lacks a polynomial kernel when parameterized by $k+d$ unless $\classNP \subseteq \classCoNP/{\rm poly}$.

Shifting the focus to the stronger maximum vertex degree $\Delta$ parameterization, \Cref{thm:kern-degree} establishes improved kernelization bounds. Specifically, \probStab admits a polynomial kernel on frameworks that outputs an equivalent framework with a graph on at most $k^2\Delta$ vertices. 

\item  When it comes to perfect graphs, there is no hope of polynomial or even parameterized algorithms with parameter $k$: 
\probMIS is already known to be \classNP-complete and \classW{1}-complete when parameterized by $k$ on bipartite graphs  by the straightforward reduction from the dual \textsc{Mulitcolored Biclique} problem~\cite{CyganFKLMPPS15}.
Also, the unconditional lower bound from \Cref{thm:lb-uncond} holds for bipartite and chordal graphs if the input matroids are given by the independence oracles.

Interestingly, it is still possible to design \classFPT algorithms for frameworks with chordal graphs when the input matroids are given by their representations. 
In \Cref{thm:FPT-chordal}, we show that 
\probStab can be solved in $2^{\Oh(k)}\cdot \|A\|^{\Oh(1)}$ time by a one-sided error Monte Carlo algorithm with false negatives on frameworks with chordal graphs and linear matroids given by their representations $A$. When it concerns kernelization, 
\Cref{thm:nokern-chord} shows that 
\probStab on frameworks with chordal graphs and partition matroids does not admit a polynomial kernel when parameterized by $k$ unless $\classNP\subseteq\classCoNP/{\rm poly}$.

\end{itemize}

\section{Preliminaries}\label{sec:prelim} 
In this section, we introduce the basic notation used throughout the paper and provide some auxiliary results.

\paragraph{Graphs.} We use standard graph-theoretic terminology and refer to the textbook of Diestel~\cite{Diestel12} for missing notions. 
 We consider only finite undirected graphs.  For a graph $G$,  $V(G)$ and $E(G)$ are used to denote its vertex and edge sets, respectively. Throughout the paper, we use $n$ to denote the number of vertices if it does not create confusion.  
 For a graph $G$ and a subset $X\subseteq V(G)$ of vertices, we write $G[X]$ to denote the subgraph of $G$ induced by $X$. 
 We denote by  $G-X$ the graph obtained from $G$ by the deletion of every vertex of $X$ (together with incident edges).   
For $v\in V(G)$, we use $N_G(v)$ to denote the \emph{(open) neighborhood} of $v$, that is, the set of vertices of $G$ that are adjacent to $v$; $N_G[v]=N_G(v)\cup\{v\}$ is the \emph{closed neighborhood} of $v$. For a set of vertices $X$, $N_G(X)=\big(\bigcup_{v\in X}N_G(v)\big)\setminus X$ and $N_G[X]=\bigcup_{v\in X}N_G[v]$. We use $d_G(v)=|N_G(v)|$ to denote the \emph{degree} of $v$; $\delta(G)$ and $\Delta(G)$ denote the minimum and maximum degree of a vertex in $G$, respectively. For a nonnegative integer $d$, $G$ is \emph{$d$-degenerate} if for every subgraph $H$ of $G$, $\delta(H)\leq d$.  Equivalently, a  graph $G$ is $d$-degenerate if there is an ordering $v_1,\ldots,v_n$ of the vertices of $G$, called \emph{elimination ordering}, such that $d_{G_i}(v_i)\leq d$ for every $i\in\{1,\ldots,n\}$ where $G_i=G[\{v_i,\ldots,v_n\}]$. Given a $d$-degenerate graph $G$, the elimination ordering can be computed in linear time~\cite{MatulaB83}. 
The \emph{degeneracy} of $G$ is the minimum $d$ such that $G$ is $d$-degenerate. We remind that a graph $G$ is \emph{bipartite} if its vertex set can be partitioned into two sets $V_1$ and $V_2$ in such a way that each edge has one endpoint in $V_1$ and one endpoint in $V_2$. A graph $G$ is \emph{chordal} if it has no induced cycles on at least four vertices. 
A graph $G$ is said to be \emph{claw-free} if it does not contain the \emph{claw} graph $K_{1,3}$ as an induced subgraph. An independent set of three vertices such that each pair can be joined by a path avoiding the neighborhood of the third is called an \emph{asteroidal triple} (AT). A graph is \emph{AT-free} if it does not contain asteroidal triples.

\paragraph{Matroids.}  We refer to the textbook of Oxley~\cite{Oxley11} for the introduction to Matroid Theory. Here we only briefly introduce the most important notions.
\begin{definition}\label{def:matroid}
A pair $\cM=(V,\cI)$, where $V$ is a \emph{ground set} and $\cI$ is a family of subsets, called \emph{independent sets of $\cM$}, is a \emph{matroid} if it satisfies the following conditions, called \emph{independence axioms}:
\begin{itemize}
\item[(I1)]  $\emptyset\in \cI$, 
\item[(I2)]  if $X \subseteq Y $ and $Y\in \cI$ then $X\in\cI$, 
\item[(I3)] if $X,Y  \in \cI$  and $ |X| < |Y| $, then there is $v\in  Y \setminus X $  such that $X\cup\{v\} \in \cI$.
\end{itemize}
\end{definition}
We use $V(\cM)$ and $\cI(\cM)$ to denote the ground set and the family of independent sets of $\cM$, respectively, unless $\cM$ is clear from the context.   
An inclusion-maximal set of $\cI$ is called a \emph{base}; it is well-known that all bases of $\cM$ have the same cardinality.
A function $r\colon 2^V\rightarrow \mathbb{Z}_{\geq 0}$  such that for every $X\subseteq V$,
\begin{equation*}
r(X)=\max\{|Y|\colon Y\subseteq X\text{ and }Y\in \cI\}
\end{equation*}
is called the \emph{rank function} of $\cM$. The \emph{rank of $\cM$}, denoted $r(\cM)$, is $r(V)$; equivalently, the rank of $\cM$ is the size of any base of $\cM$.  
Let us remind that a set $X\subseteq V$ is independent if and only if $r(X)=|X|$. 
The \emph{closure} of a set $X$ is the set $\cl(X)=\{v\in V\colon r(X\cup\{v\})=r(X)\}$.
The matroid $\cM'=(V\setminus X,\cI')$, where $\cI'=\{Y\in \cI\colon Y\subseteq V\setminus X\}$, is said to be obtained from $\cM$ by the \emph{deletion} of $X$. The \emph{restriction}  of $\cM$ to $X\subseteq V$ is the matroid obtained from $\cM$ by the deletion of $V\setminus X$.  
If $X$ is an independent set then the matroid $\cM''=(V\setminus X,\cI'')$, where 
$\cI''=\{Y\subseteq V\setminus X \colon Y\cup X\in\cI\}$, is the \emph{contraction} of  $\cM$ by $X$. 
For a positive integer $k$, the \emph{$k$-truncation} of $\cM=(V,\cI)$ is the matroid $\cM'$ with the same ground set $V$ such that $X\subseteq V$ is independent in $\cM'$ if and only if $X\in \cI$ and $|X|\leq k$. Because in \probStab, we are interested only in independent sets of size at most $k$, we assume throughout our paper that the rank of the input matroids is upper bounded by $k$. Otherwise, we replace $\cM$ by its $k$-truncation. 

In our paper,   we assume in the majority of our algorithmic results that the input matroids in instances of \probStab are given by independence oracles. An \emph{independence oracle} for $\cM$ takes as its input a set $X\subseteq V$ and correctly returns either {\sf yes} or {\sf no}  in unit time depending on whether  $X$ is independent or not. 
We assume that the memory used to store oracles does not contribute to the input size; this is important for our kernelization results. 
Notice that given an independence oracle, we can greedily construct an inclusion-maximal independent subset of $X$ and this can be done in $\Oh(|X|)$ time.  Note also that the oracle for $\cM$ can be trivially transformed to an oracle for the $k$-truncation of $\cM$.

Our computational lower bounds, except the unconditional bound in \Cref{thm:lb-uncond}, are established for partition matroids. The \emph{partition} matroid for a given partition $\{V_1,\ldots,V_\ell\}$ of $V$ is the matroid with the ground set $V$ such that a set $X\subseteq V$ is independent if and only if $|X\cap V_i|\leq 1$ for each $i\in\{1,\ldots,\ell\}$ (in a more general setting, it is required that $|V\cap X_i|\leq d_i$ where $d_1,\ldots,d_\ell$ are some constant but we only consider the case $d_1=\dots=d_\ell=1$). 

Matroids also could be given by their representations. Let  $\cM=(V,\cI)$ be a matroid and let $\mathbb{F}$ be a field. An $r\times n$-matrix $A$ is a \emph{representation of $\cM$ over $\mathbb{F}$} if there is a bijective correspondence $f$ between $V$ and the set of columns of $A$ such that for every $X\subseteq V$, $X\in \cI$ if and only if the set of columns $f(X)$ consists of linearly independent vectors of $\mathbb{F}^r$. Equivalently, $A$ is a representation of $M$ if $M$ is isomorphic to the \emph{column} matroid of $A$, that is, the matroid whose ground set is the set of columns of the matrix and the independence of a set of columns is defined as the linear independence.   
If $\cM$ has a  such a representation, then $\cM$ is \emph{representable} over $\mathbb{F}$ and it is also said
$M$ is a \emph{linear} (or \emph{$\mathbb{F}$-linear}) matroid. 

\paragraph{Parameterized Complexity.} We refer to the books of Cygan et al.~\cite{CyganFKLMPPS15} and Fomin et al.~\cite{fomin2019kernelization}
for an introduction to the area. Here we only briefly mention the notions that are most important to state our results.  A \emph{parameterized problem} is a language $L\subseteq\Sigma^*\times\mathbb{N}$  where $\Sigma^*$ is a set of strings over a finite alphabet $\Sigma$. An input of a parameterized problem is a pair $(x,k)$ where $x$ is a string over $\Sigma$ and $k\in \mathbb{N}$ is a \emph{parameter}. 
A parameterized problem is \emph{fixed-parameter tractable} (or \classFPT) if it can be solved in time $f(k)\cdot |x|^{\mathcal{O}(1)}$ for some computable function~$f$.  
The complexity class \classFPT contains all fixed-parameter tractable parameterized problems.
A \emph{kernelization algorithm} or \emph{kernel} for a parameterized problem $L$ is a polynomial-time algorithm that takes as its input an instance $(x,k)$ of $L$ and returns an instance $(x',k')$ of the same problem such that (i) $(x,k)\in L$ if and only if $(x',k')\in L$ and (ii) $|x'|+k'\leq f(k)$ for some computable function $f\colon \mathbb{N}\rightarrow \mathbb{N}$. The function $f$ is the \emph{size} of the kernel; a kernel is \emph{polynomial} if $f$ is a polynomial. While every decidable parameterized problem is \classFPT if and only if the problem admits a kernel, it is unlikely that all \classFPT problems have polynomial kernels.  In particular, the \emph{cross-composition} technique proposed by Bodlaender, Jansen, and Kratsch~\cite{BodlaenderJK14} could be used to prove that a certain parameterized problem does not admit a polynomial kernel unless $\classNP\subseteq\classCoNP/{\rm poly}$.

We conclude the section by defining \probMIS. 

 \defproblema{\probMIS}%
{A graph $G$ and a partition $\{V_1,\ldots,V_k\}$ of $V(G)$ into $k$ sets, called color classes.}%
{Decide whether $G$ has a stable set $S$ of size  $k$ such that $|S\cap V_i|=1$ for each $i\in\{1,\ldots,k\}$.
}

As mentioned,  \probMIS is a special case of \probStab for partition matroids where $k$ is the number of subsets in the partition defining the input matroid.

\section{Unconditional computational lower bound}\label{sec:lower-uncond}
Because \probStab generalizes the classical \probIS problem, \probStab is \classNP-complete~\cite{GareyJ79} and \classW{1}-complete~\cite{DowneyF13}. However, when the input matroids are given by their independence oracles, we obtain an unconditional computational lower bound. 
Moreover, we show that the lower bound holds for several graph classes for which the classical \probIS problem can be solved in polynomial time. For this, we remind that \probIS is polynomial on claw-free and AT-free graphs by the results of Minty~\cite{Minty80} and  Broersma et al.~\cite{BroersmaKKM99}, respectively. 

\begin{theorem}\label{thm:lb-uncond}
There is no algorithm solving \probStab for frameworks with matroids represented by the independence oracles using $f(k)\cdot n^{o(k)}$ oracle calls for any computable function $f$.
Furthermore, the bound holds for bipartite, chordal, claw-free, and  AT-free graphs.
\end{theorem}

\begin{proof}
First, we show the bound for claw-free and AT-free and then explain how to modify the proof for other graph classes.

Let $p$ and $q$ be positive integers. We define the graph $G_{p,q}$ as the disjoint union of $G_i$ constructed as follows for each $i\in\{1,\ldots,p\}$.
\begin{itemize}
\item For each $j\in\{1,\ldots,q\}$, construct two vertices $a_{i,j}$ and $b_{i,j}$; set $A_i=\{a_{i,1},\ldots,a_{i,q}\}$ and  $B_i=\{b_{i,1},\ldots,b_{i,q}\}$.
\item Make $A_i$ and $B_i$ cliques.
\item For each $j\in\{1,\ldots,q\}$ and for all distinct $h,j\in\{1,\ldots,q\}$, make $a_{i,h}$ and $b_{i,j}$ adjacent. 
\end{itemize}
Equivalently, each $G_i$ is obtained by deleting a perfect matching from the complete graph $K_{2q}$.
By the construction, $G_{p,q}$ is both claw-free and AT-free
and has $2pq$ vertices. 
Consider a family of indices $j_1,\ldots,j_p\in\{1,\ldots,q\}$ and set 
$W=\bigcup_{i=1}^p\{a_{i,j_i},b_{i,j_i}\}$.
We define the matroid $\cM_W$ with the ground set $V(G_{p,q})$ as follows for $k=2p$:
\begin{itemize}
\item  Each set $X\subseteq V(G_{p,q})$ of size at most $k-1$ is independent and any set of size at least $k+1$ is not independent.
\item A set $X\subseteq V(G_{p,q})$ of size $k$ is independent if and only if either $X=W$ or there is $i\in\{1,\ldots,p\}$ such that
$|A_i\cap X|\geq 2$ or  $|B_i\cap X|\geq 2$ or there are distinct $h,j\in\{1,\ldots,q\}$ such that 
$a_{i,h},b_{i,j}\in X$.
\end{itemize}
Denote by $\cI_W$ the constructed family of independent sets.  We will now show that $\cM_W$ is indeed a matroid.

\begin{clm}\label{cl:W-matr}
$\cM_W=(V(G_{p,q}),\cI_W)$ is a matroid.
\end{clm}

\begin{proof}[Proof of \Cref{cl:W-matr}]
We have to verify that $\cI_W$ satisfies the independence axioms (I1)--(I3). The axioms (I1) and (I2) for $\cI_W$ follow directly from the definition of $\cI_W$. To establish (I3), consider arbitrary $X,Y\in \cI_W$ such that $|X|<|Y|$. If $|X|<k-1$ then for any $v\in Y\setminus X$, $Z=X\cup \{v\}\in \cI_W$ because $|Z|\leq k-1$.

 Suppose $|X|=k-1$ and $|Y|=k$. 
 If  there is $i\in\{1,\ldots,p\}$ such that $|A_i\cap X|\geq 2$ or  $|B_i\cap X|\geq 2$ or there are distinct $h,j\in\{1,\ldots,q\}$ such that 
$a_{i,h},b_{i,j}\in X$ then for any $v\in Y\setminus X$, the set $Z=X\cup \{v\}$ has the same property and, therefore, $Z\in\cI_W$.
Assume that this is not the case. By the construction of $G_{p,q}$, we have that for each $i\in\{1,\ldots,p\}$, $|X\cap A_i|\leq 1$ and $|X\cap B_i|\leq 1$, and, furthermore, 
there is $j\in\{1,\ldots,q\}$ such that $X\cap (A_i\cup B_i)\subseteq \{a_{i,j},b_{i,j}\}$. Because $|X|=k-1$, we can assume without loss of generality that there are indices $h_1,\ldots,h_p\in \{1,\ldots,q\}$ such that $X\cap (A_i\cup B_i)= \{a_{i,h_i},b_{i,h_i}\}$ for $i\in\{1,\ldots,p-1\}$ and $X\cap (A_p\cup B_p)=\{a_{p,h_p}\}$. Recall that $W=\bigcup_{i=1}^p\{a_{i,j_i},b_{i,j_i}\}$ for $j_1,\ldots,j_p\in\{1,\ldots,q\}$. If there is $v\in Y\setminus X$ such that $v\neq b_{p,j_p}$ then consider $Z=X\cup\{v\}$.  We have that there is $i\in\{1,\ldots,p\}$ such that
$|A_i\cap Z|\geq 2$ or  $|B_i\cap Z|\geq 2$ or there are distinct $h,j\in\{1,\ldots,q\}$ such that 
$a_{i,h},b_{i,j}\in Z$, that is,   $Z\in \cI_W$. 
Now we assume that $Y\setminus X=\{b_{p,j_p}\}$. Then $Y=W$ and we can take $v=b_{p,j_p}$. We obtain that $X\cup\{v\}=Y\in \cI_W$. This concludes the proof.
\end{proof}

We show the following lower bound for the number of oracle queries for frameworks $(G_{p,q},\cM_W)$.

\begin{clm}\label{cl:W-lower}
Solving \probStab for instances $(G_{p,q},\cM_W,k)$ with the matroids $\cM_W$ defined by the independence oracle for an (unknown) stable set $W$ of $G_{p,q}$ of size $k$ demands at least $q^p-1$ oracle queries.   
\end{clm}

\begin{proof}[Proof of \Cref{cl:W-lower}] 
Notice that every stable set of $X$ of size $k$ contains exactly two vertices of each $G_i$ and, moreover, there is $j\in\{1,\ldots,q\}$ such that $X\cap V(G_i)=\{a_{i,j},b_{i,j}\}$. Because the only stable set of this structure that is independent with respect to $\cM_W$ is $W$,  
 the task of \probStab  boils down to finding an unknown stable set $W$ of $G_{p,q}$ of size $k$ using oracle queries. Querying the oracle for sets $X$ of size at most $k-1$ or at least $k+1$ does not provide any information about $W$. Also, querying the oracle for $X$ of size $k$ with the property that there is $i\in\{1,\ldots,p\}$ such that
$|A_i\cap X|\geq 2$ or  $|B_i\cap X|\geq 2$ or there are distinct $h,j\in\{1,\ldots,q\}$ such that 
$a_{i,h},b_{i,j}\in X$  also does not give any information because all these are independent. Hence, we can assume that the oracle is queried only for sets $X$ of size $k$ with the property that for each $i\in\{1,\ldots,p\}$, there is $j\in\{1,\ldots,q\}$ such that $X\cap V(G_i)=\{a_{i,j},b_{i,j}\}$, that, is the oracle is queried for 
stable sets of size $k$. The graph $G_{p,q}$ has $q^p$ such sets. Suppose that the oracle is queried for at most $q^p-2$ stable sets of size $k$ with the answer {\sf no}. Then there are two distinct stable sets $W$ and $W'$ of size $k$ such that the oracle was queried neither for $W$ nor $W'$. The previous queries do not help to distinguish between $W$ and $W'$. Hence, at least one more query is needed. This proves the claim. 
\end{proof}

Now, we are ready to prove the claim of the theorem. Suppose that there is an algorithm $\mathcal{A}$ solving \probStab with at most $f(k)\cdot n^{g(k)}$ oracle calls for computable functions $f$ and $g$ such that $g(k)=o(k)$. Without loss of generality, we assume that $f$ and $g$ are monotone non-decreasing functions.
Because $g(k)=o(k)$, there is a positive integer $K$ such that $g(k)<k/2$ for all $k\geq K$. Then for 
each $k\geq K$, there is a positive integer $N_k$ such that for every $n\geq N_k$, $(f(k)\cdot n^{g(k)}+1)k^{k/2}<n^{k/2}$. 

Consider instances $(G_{p,q},\cM_W,k)$ for even $k\geq K$ where $p=k/2$ and $q\geq N_k/k$. We have that $k=2p$ and $n=2pq$. Then $\mathcal{A}$ applied to such instances would use at most $f(k)\cdot n^{g(k)}<\big(\frac{n}{k}\big)^{k/2}-1=q^p-1$ oracle queries contradicting \Cref{cl:W-lower}. This completes the proof for claw and AT-free graphs. 

\medskip
Now we sketch the proof of \Cref{thm:lb-uncond} for bipartite graphs. For positive integers $p$ and $q$, we  define $H_{p,q}$ as the disjoint union of the graphs $H_i$ for $i\in\{1,\ldots,p\}$ constructed as follows.
\begin{itemize}
\item For each $j\in\{1,\ldots,q\}$, construct three vertices $a_{i,j}$, $b_{i,j}$, and $c_{i,j}$; set $A_i=\{a_{i,1},\ldots,a_{i,q}\}$,  $B_i=\{b_{i,1},\ldots,b_{i,q}\}$, and 
 $C_i=\{c_{i,1},\ldots,c_{i,q}\}$. 
\item For each $j\in\{1,\ldots,q\}$, make $a_{i,j}$ and $b_{i,j}$ adjacent to every $c_{i,h}$ for $h\in\{1,\ldots,q\}$ such that $h\neq j$. 
\end{itemize}
Notice that $H_{p,q}$ is a bipartite graph with $3pq$ vertices. We define $R=\bigcup_{i=1}^p(A_i\cup B_i)$.
Consider a family of indices $j_1,\ldots,j_p\in\{1,\ldots,q\}$ and set 
$W=\bigcup_{i=1}^p\{a_{i,j_i},b_{i,j_i}\}$. Note that $W$ is a stable set of $H_{p,q}$ of size $2p$.
We define the matroid $\cM_W$ with the ground set $V(H_{p,q})$ by setting a set $X\subseteq V(H_{p,q})$ to be independent if and only if 
\begin{itemize}
\item for each $i\in\{1,\ldots,p\}$, $|C_i\cap X|\leq 1$ and
\item it holds that  
\begin{itemize}
\item either $X\cap R=W$,
\item or $|X\cap R|<2p$,
\item or $|X\cap R|=2p$ and there is $i\in\{1,\ldots,p\}$ such that $|A_i\cap X|\geq 2$ or $|B_i\cap X|\geq 2$ or there  
are distinct $h,j\in\{1,\ldots,q\}$ such that $a_{i,h},b_{i,j}\in X$. 
\end{itemize} 
\end{itemize}
We denote by $\cI_W$ the constructed family of independent sets and prove that $\cM_W$ is a matroid.

\begin{clm}\label{cl:W-matr-two}
$\cM_W=(V(H_{p,q}),\cI_W)$ is a matroid.
\end{clm}

\begin{proof}[Proof of \Cref{cl:W-matr-two}]
Let $S=\bigcup_{i=1}^pC_i$ and consider $\cM_1=(S,\cI_1)$ where $\cI_1$ is the set of all $X\subseteq S$ such that $|X\cap C_i|\leq 1$ for $i\in\{1,\ldots,p\}$. Clearly, $\cM_1$ is a partition matroid. Now consider
$\cM_2=(R,\cI_2)$ where $\cI_2$ consists of sets $X\subseteq D$ such that either $X=W$, or $|X|<2p$, or $|X|=2p$ and there is $i\in\{1,\ldots,p\}$ such that $|A_i\cap X|\geq 2$ or $|B_i\cap X|\geq 2$ or there  
are distinct $h,j\in\{1,\ldots,q\}$ such that $a_{i,h},b_{i,j}\in X$. 
We observe that $\cM_2$ is a matroid and the proof of this claim is identical to the proof of \Cref{cl:W-matr}. To complete the proof, it remains to note that $\cM_W=\cM_1\cup \cM_2$, that is, a set $X\in \cI_W$ if and only if 
$X=Y\cup Z$ for $Y\in \cI_1$ and $Z\in \cI_2$. This implies that $\cM_W$ is a matroid~\cite{Oxley11}.
\end{proof}

We consider instances $(H_{p,q},\cM_W,k)$ of \probStab with the matroid $\cM_W$ defined by the independence oracle for an (unknown) $W$ and $k=3p$. By the definition of $\cM_W$, any stable set $X$ of $H_{p,q}$ of size $k$ that is independent with respect to $\cM_W$ has the property that $|X\cap C_i|=1$ for every $i\in\{1,\ldots,p\}$. The construction of $H_{p,q}$ implies that if  $c_{i,j}\in X\cap C_i$ then $X\cap(A_i\cup B_i)\subseteq \{a_{i,j},b_{i,j}\}$. Because $|X|=k=3p$, we obtain that  $X\cap(A_i\cup B_i)=\{a_{i,j},b_{i,j}\}$. Then by the construction of $\cM_W$, we obtain that $X=\bigcup_{i=1}^p\{a_{i,j_i},b_{i,j_i},c_{i,j_i}\}$ where $W=\bigcup_{i=1}^p\{a_{i,j_i},b_{i,j_i}\}$, that is, $X$ is uniquely defined by $W$. In the same way as in \Cref{cl:W-lower} we obtain that solving \probStab for instances $(H_{p,q},\cM_W,k)$ with the matroids $\cM_W$ defined by the independence oracle for an (unknown)  $W$ demands at least $q^p-1$ oracle queries. Similarly to the case of claw and AT-free graphs, we conclude that the existence of an algorithm for \probStab using $f(k)\cdot n^{o(k)}$ oracle calls would lead to a contradiction.  This finishes the proof for bipartite graphs.

For chordal graphs, we modify the construction of $H_{p,q}$ by making each $C_i$ a clique. Then $H_{p,q}$ becomes chordal but we can apply the same arguments to show 
  that solving \probStab for instances $(H_{p,q},\cM_W,k)$ with the matroids $\cM_W$ defined by the independence oracle for an (unknown)  $W$ demands at least $q^p-1$ oracle queries. This completes the proof. 
\end{proof}

\section{\probStab on sparse frameworks}\label{sec:deg}
In this section, we consider \probStab for graphs of bounded maximum degree and graphs of bounded degeneracy. 
First, we observe that the problem is \classFPT when parameterized by the solution size and the degeneracy by giving a recursive branching algorithm. 

\begin{theorem}\label{thm:fpt-degen}
\probStab can be solved in $\Oh((d+1)^k\cdot n)$ time on frameworks with $d$-degenerate input graphs.
\end{theorem}

\begin{proof}
The algorithm is based on the following observation. Let $(G,\cM)$ be a framework such that for every $v\in V(G)$, $\{v\}\in \cI$. Then there is a stable set $X$ of $G$ that is independent with respect to $\cM$ whose size is maximum such that  $X\cap N_G[v]\neq\emptyset$.  To see this, let $X$ be a stable set that is also independent in $\cM$ and such that $X\cap N_G[v]=\emptyset$. Because $\{v\}$ and $X$ are independent, there is $Y\subseteq X$ of size $|X|-1$ such that $Z=Y\cup \{v\}$ is independent. Because $N_G(v)\cap Z=\emptyset$ and $Y$ is a stable set, $Z$ is a stable set. Thus, set $Z$ of size $|X|$ is stable in $G$ and is independent in $\cM$. This proves the observation.

Consider an instance $(G,\cM,k)$ of \probStab. Because $G$ is a $d$-degenerate graph, there is an elimination ordering $v_1,\ldots,v_n$ of the vertices of $G$, that is, $d_{G_i}(v_i)\leq d$ for every $i\in\{1,\ldots,n\}$ where $G_i=G[\{v_i,\ldots,v_n\}]$. Recall that such an ordering can be computed in linear time~\cite{MatulaB83}.

If there is $v\in V(G)$ such that $\{v\}\notin \cI$, then we delete $v$ from the framework as such vertices are trivially irrelevant. From now on, we assume that $\{v\}\in\cI$ for any $v\in V(G)$.  If $k=0$, then $\emptyset$ is a solution, and we return {\sf yes} and stop. If $k\geq 1$ but $V(G)=\emptyset$, then we conclude that the answer is {\sf no} and stop. We can assume that $V(G)\neq\emptyset$ and $k\geq 1$.

Let $u$ be the first vertex in the elimination ordering. Clearly, $d_G(u)\leq d$. We branch on at most $d+1$ instances $(G-v,\cM/v,k-1)$ for $v\in N_G[u]$, where $\cM/v$ is the contraction of $\cM$ by $\{v\}$. By our observation,  $(G,\cM,k)$ is a yes-instance of \probStab if and only if at least one of the instances $(G-v,\cM/v,k-1)$ is a yes-instance. 

In each step, we have at most $d+1$ branches and the depth of the search tree is at most $k$.  Note that we do not need to recompute the elimination ordering when a vertex is deleted; instead, we just delete the vertex from the already constructed ordering. This means we can use the ordering constructed for the original input instance.
 Thus, the total running time is $\Oh((d+1)^k\cdot n)$. This concludes the proof.
\end{proof}

For bounded degree graphs, we prove that \probStab has a polynomial kernel when parameterized by $k$ and the maximum degree.

\begin{theorem}\label{thm:kern-degree} 
\probStab admits a polynomial kernel on frameworks with graphs of maximum degree at most $\Delta$ such that the output instance contains a graph with at most $k^2\Delta$ vertices.
\end{theorem}

\begin{proof}
Let $(G,\cM,k)$ be an instance of \probStab with $\Delta(G)\leq \Delta$. Recall that by our assumption, $r(\cM)\leq k$. If $r(\cM)<k$ then $(G,\cM,k)$ is a no-instance. In this case, our kernelization algorithm returns a trivial no-instance of constant size and stops. Now we can assume that $r(\cM)=k$. If $k=0$  then we return a trivial yes-instance as $\emptyset$ is a solution. If $\Delta=0$, then any base of $\cM$ is a solution, and we return a trivial yes-instance.
Now we can assume that $k\geq 1$ and $\Delta\geq 1$.  

We set $W_0=\emptyset$. Then for $i=1,\ldots,\ell$ where $\ell=k\Delta$, we greedily  select a maximum-size independent set $W_i\subseteq  V(G)\setminus \big(\bigcup_{j=0}^{i-1} W_{j}\big)$. Our kernelization algorithms returns the instance $(G',\cM',k)$ where $G'=G[\bigcup_{i=1}^\ell W_i]$ and $\cM'$ is the restriction of $\cM$ to $V(G')$. It is straightforward to see that
$|V(G')|\leq k^2\Delta$ as $|W_i|\leq r(\cM)=k$ and the new instance can be constructed in polynomial time. 
We claim that $(G,\cM,k)$ is a yes-instance of \probStab if and only if $(G',\cM',k)$ is a yes-instance. 

Because $G'$ is an induced subgraph of $G$, any stable set of $G'$ is a stable set of $G$. This immediately implies that if $(G',\cM',k)$ is a yes-instance then any solution to $(G',\cM',k)$ is a solution to $(G,\cM,k)$ and, thus, $(G,\cM,k)$ is a yes-instance.  Suppose that $(G,\cM,k)$ is a yes-instance. It means that  $G$ contains a stable set of size $k$ independent in $\cM$. We show that there is a stable set  $X\subseteq V(G')$  of $G$ of size $k$ that is independent with respect to $\cM$. 

To show this, let $X$ be a stable set of size $k$ that is independent in $\cM$ with the maximum number of vertices in $V(G')$. For the sake of contradiction, assume that there is $u\in X\setminus V(G')$.  We define $Y=X\setminus \{u\}$. Consider the set $W_i$ for some $i\in\{1,\ldots,\ell\}$. By the construction of the set, we have that $u\in \cl(W_i)$. Then it holds that $r(Y\cup W_i)\geq r(X)$. 
This implies that there is $w_i\in W_i$ such that $r(Y\cup\{w_i\})=r(X)=k$. Because this property holds for arbitrary $i\in\{1,\ldots,\ell\}$, we obtains that there are $\ell=k\Delta$ vertices $w_1,\ldots,w_\ell\in V(G')$ such that for any $i\in\{1,\ldots,\ell\}$, $r(Y\cup \{w_i\})=k$. Notice that $w_i\notin Y$ for $i\in \{1,\ldots,\ell\}$ and $|N_G(Y)|\leq (k-1)\Delta$. Therefore, there is $i\in\{1,\ldots,\ell\}$ such that $w_i$ is not adjacent to any vertex of $Y$. Then $Z=Y\cup\{w_i\}$ is a stable set of $G$. However, $|Z\cap V(G')|>|X\cup V(G')|$ contradicting the choice of $X$. This proves that there is a stable set  $X\subseteq V(G')$  of $G$ of size $k$ that is independent in $\cM$. Then $X$ is a solution to $(G',\cM',k)$, that is, $(G',\cM',k)$ is a yes-instance. This concludes the proof.
\end{proof}

\Cref{thm:kern-degree} is handy for kernelization with parameter $k$ when the degeneracy of the graph in a framework is a constant.

\begin{theorem}\label{thm:kern-degen} 
For every integer $d\geq 0$, \probStab admits a polynomial kernel with running time $n^{\Oh(d)}$ on frameworks with graphs of degeneracy at most $d$ such that the output instance contains a graph with $dk^{\Oh(d)}$ vertices. 
\end{theorem}

\begin{proof}
Let $(G,\cM,k)$ be an instance of \probStab where the degeneracy of $G$ is at most $d$. We assume without loss of generality that $r(\cM)=k$. Otherwise, if $r(\cM)<k$, then $(G,\cM,k)$ is a no-instance, and we can return a trivial no-instance of constant size and stops. If $d=0$, then $G$ is an edgeless graph, and any set of vertices forming a base of $\cM$ is a stable set of size $k$ that is independent with respect to $\cM$, that is, $(G,\cM,k)$ is a yes-instance. Then we return a trivial yes-instance and stop. From now on, we assume that $d\geq 1$. Also, we assume that $k\geq 2$. Otherwise, if 
$k=0$, the empty set is a trivial solution.  If $k=1$ then because $r(\cM)=k\geq 1$, there is a vertex $v$ such that $\{v\}\in \cI(\cM)$ and $\{v\}$ is an independent set of size $k$. In both cases, we return a trivial yes-instance and stop.

Since  $G$ is a $d$-degenerate graph, it admits an elimination ordering $v_1,\ldots,v_n$ of the vertices of $G$, that is, $d_{G_i}(v_i)\leq d$ for every $i\in\{1,\ldots,n\}$ where $G_i=G[\{v_i,\ldots,v_n\}]$. Recall that such an ordering can be computed in linear time~\cite{MatulaB83}. For a set of vertices $X\subseteq V(G)$, we use $F(X)$ to denote the set of common neighbors of the vertices of $X$ that occur before the vertices of $X$ in the elimination ordering. Note that because $G$ is a $d$-degenerate graph, $F(X)=\emptyset$ if $|X|>d$. 
For an integer $i\geq 1$, $f_i(G)=\max\{|F(X)|\colon X\subseteq V(G)\text{ and }|X|=i\}$. Clearly, $f_i(G)=0$ if $i>d$.

For each $h=d,\ldots,1$, we apply the following reduction rule starting with $h=d$. Whenever the rule deletes some vertices, we do not recompute the elimination ordering; instead, we use the induced ordering obtained from the original one by vertex deletions.
\begin{redrule}\label{red:main}
Set $d_h=d+f_{h+1}(G)$. For each $X\subseteq V(G)$ such that $|X|=h$, do the following:
\begin{itemize}
\item[(i)]  set $W_0=\emptyset$,
\item[(ii)] for $i=1,\ldots,\ell$ where $\ell=kd_h$, greedily  select a maximum-size independent set $W_i\subseteq  F(X)\setminus \big(\bigcup_{j=0}^{i-1} W_{j}\big)$,
\item[(iii)] delete the vertices of $D=F(X)\setminus \big(\bigcup_{i=1}^\ell W_i\big)$ and restrict $\cM$ to $V(G)\setminus D$.
\end{itemize}
\end{redrule}

It is easy to see that the rule can be applied in $n^{\Oh(d)}$ time. We show that the rule is \emph{safe}, that is, it returns an equivalent instance of the problem.

\begin{clm}\label{cl:safe}
\Cref{red:main} is safe.
\end{clm}

\begin{proof}[Proof of \Cref{cl:safe}]
Let $X\subseteq V(G)$ be of size $h$. Denote by $G'$ the graph obtained from $G$ by applying steps (i)--(iii) for $X$ and let $\cM$ be the restriction of $\cM$ to $V(G)\setminus D$. 
We prove that $(G,\cM,k)$ is a yes-instance of \probStab if and only if $(G',\cM',k)$ is a yes-instance. Clearly, this is sufficient for the proof of the claim. Since $G'$ is an induced subgraph of $G$, any solution to $(G',\cM',k)$ is a solution to $(G,\cM,k)$. Thus, if $(G',\cM',k)$ is a yes-instance then the same holds for $(G,\cM,k)$. 
Hence, it remains to show that if $(G,\cM,k)$ is a yes-instance then $(G',\cM',k)$ is a yes-instance as well.

We use the following axillary observation: for every $v\in V(G)\setminus X$, $|N_G(v)\cap F(X)|\leq d_h$. To see this, consider $v\in V(G)\setminus X$, and denote by $L$ and $R$ the sets of vertices of $F(X)$ that are prior and after $v$, respectively,  in the elimination ordering.  By the definition of an elimination ordering, $|N_G(v)\cap R|\leq d$.  For $N_G(v)\cap L$, we have that $N_G(v)\cap L\subseteq F(X\cup\{v\})$. Then $|N_G(v)\cap L|\leq |F(X\cup\{v\})|\leq f_{h+1}$. We conclude that 
$|N_G(v)\cap F(X)|=|N_G(v)\cap L|+|N_G(v)\cap R|\leq d+f_{h+1}=d_h$. This proves the observation. 

Suppose that $G$ has a stable set $Y$ of size $k$ that is independent with respect to $\cM$. Among all these sets, we select $Y$ such that $Y\cap D$ has the minimum size. We claim that $Y\cap D=\emptyset$.  The proof is by contradiction and is similar to the proof of \Cref{thm:kern-degree}. Assume that there is $u\in Y\cap D$ and let $Z=Y\setminus\{u\}$.
For each $i\in\{1,\ldots,\ell\}$, $u\in \cl(W_i)$ by the construction of $W_i$. Thus, $r(Z\cup W_i)\geq r(Y)$ and for each $i\in\{1,\ldots,\ell\}$, 
 there is $w_i\in W_i$ such that $r(Z\cup\{w_i\})=r(Y)=k$. Therefore, there are $\ell$ vertices $w_1,\ldots,w_\ell\in F(X)\setminus D$ such that for any $i\in\{1,\ldots,\ell\}$, $r(Z\cup \{w_i\})=k$.  Notice that $w_i\notin Z$ for all $i\in\{1,\ldots,\ell\}$ and $Y\cap X=\emptyset$ because $u$ is adjacent to every vertex of $X$. By the above observation, we have that $|N_G(Z)\cap F(X)|\leq (k-1)d_h$. Since $\ell=kd_h>(k-1)d_h$, there is $i\in \{1,\ldots,\ell\}$ such that $w_i\notin N_G(Z)$. Then $Y'=Z\cup\{w_i\}$ is a stable set of $G$. Because $Y'$ is independent with respect to $\cM$ and $u\notin D$, this leads to a contradiction with the choice of $Y$. We conclude that there is a stable set $Y$ of $G$ of size $k$ that is independent with respect to $\cM$ such that $Y\cap D=\emptyset$. This means that $Y$ is a solution to $(G',\cM',k)$, that is, $(G',\cM',k)$ is a yes-instance of \probStab. This concludes the proof. 
\end{proof}

Denote by $(G',\cM',k)$ the instance of \probStab obtained after applying \Cref{red:main}. We prove that the maximum degree of $G'$ is bounded. 

\begin{clm}\label{cl:max-deg}
$\Delta(G')\leq dk^{2d+1}$.
\end{clm}

\begin{proof}[Proof of \Cref{cl:max-deg}]
For $i\in\{1,\ldots,d\}$, denote by $G_i$ the graph obtained from $G$ by applying  \Cref{red:main} for $h=d,\ldots,i$. Note that $G'=G_1$. Because $r(\cM)=k$, for each set $W_j$ selected in step (ii) of \Cref{red:main}, $|W_j|\leq k$. Therefore, $|\bigcup_{j=1}^{\ell}W_j|\leq k\ell=k^2d_h$. Notice that for $h=d$, $f_{h+1}(G)=0$ and, therefore, $d_h=d$. This implies that 
$f_d(G_d)\leq k^2d$. For $i<d$, we have that $f_i(G_i)\leq k^2d_i=k^2(d+f_{i+1}(G_{i+1}))$. Therefore,
$f_i(G_i)\leq d\sum_{j=i}^dk^{2(j-i+1)}$
and, as $k\geq 2$,
\begin{equation*}
f_1(G')\leq f_1(G_1)\leq d\sum_{j=1}^dk^{2j}=d\sum_{j=0}^dk^{2j}-d=d\frac{k^{2(d+1)}-1}{k^2-1}-d\leq dk^{2d+1}-d. 
\end{equation*}
Therefore, each vertex $v$ of $G'$ has at most $dk^{2d+1}-d$ neighbors in $G'$ that are prior $v$ in the elimination ordering. Because $v$ has at most $d$ neighbors that are after $v$ in the ordering, $d_{G'}(v)\leq dk^{2d+1}$. This concludes the proof.
\end{proof}

Because the maximum degree of $G'$ is bounded, we can apply \Cref{thm:kern-degree}. Applying the kernelization algorithm from this theorem to $(G',\cM',k)$, we obtain a kernel with at most 
$dk^{2d+3}$ vertices. This concludes the proof of the theorem. 
\end{proof}

In \Cref{thm:kern-degen}, we proved that \probStab admits a polynomial kernel on $d$-degenerate graphs when $d$ is a fixed constant. We complement this result by showing that it is unlikely that the problem has a polynomial kernel when parameterized by both $k$ and $d$.

\begin{theorem}\label{thm:nokern-degen}
\probStab on frameworks with $d$-degenerate graphs and partition matroids does not admit a polynomial kernel when parameterized by $k+d$ unless $\classNP\subseteq\classCoNP/{\rm poly}$.
\end{theorem}

\begin{proof}
We use the fact that \probMIS is a special case of \probStab and show that \probMIS  does not admit a polynomial kernel when parameterized by $k+d$ unless $\classNP\subseteq\classCoNP/{\rm poly}$ where $k$ is the number of color classes.

We use cross-composition from \probMIS. 
We say that two instances $(G,\{V_1,\ldots,V_k\})$ and $(G',\{V_1',\ldots,V_{k'}'\})$ are \emph{equivalent} if $|V(G)|=|V(G')|$ and $k=k'$. Consider $t$ equivalent instances 
$(G_i,\{V_1^i,\ldots,V_k^i\})$ of \probMIS for $i\in\{1,\ldots,t\}$ where each graph has $n$ vertices. We assume that $t=2^p$ for some $p\geq 1$. Otherwise, we add 
$2^{\lceil \log t\rceil}-t$ copies of $(G_1,\{V_1^1,\ldots,V_k^1\})$ to achieve the property for $p=\lceil \log t\rceil$; note that by this operation, we may add at most $t$ instances.
Then we construct the instance $(G,\{V_1,\ldots,V_{k+p}\})$ of \probMIS as follows.
\begin{itemize}
\item Construct the disjoint union of copies of $G_1,\ldots,G_t$.
\item For each $i\in\{1,\ldots,p\}$, 
\begin{itemize}
\item construct two adjacent vertices $u_i$ and $v_i$,
\item for each $j\in\{1,\ldots,t\}$, consider the binary encoding of $j-1$ as a string $s$ with $p$ symbols and make 
 all the vertices of $G_j$ adjacent to $u_i$ if $s[i]=0$ and make them adjacent to $v_i$, otherwise, for $i\in\{1,\ldots,p\}$.
\end{itemize}
\item 
Define $k+p$ color classes $V_i=\bigcup_{j=1}^tV_i^j$ for $i\in\{1,\dots,k\}$ and $V_{k+i}=\{u_i,v_i\}$ for $i\in\{1,\ldots,p\}$.
\end{itemize}
It is straightforward to see that
the instance $(G,\{V_1,\ldots,V_{k+p}\})$ of \probMIS can be constructed in polynomial time.
We claim that  $(G,\{V_1,\ldots,V_{k+p}\})$ is a yes-instance of \probMIS if and only if there is $j\in\{1,\ldots,t\}$ such that $(G_j,\{V_1^j,\ldots,V_k^j\})$ is a yes-instance of \probMIS.

Suppose that $(G_j,\{V_1^j,\ldots,V_k^j\})$ is a yes-instance for some $j\in\{1,\ldots,t\}$. Then there is a stable set $X\subseteq V(G_j)$ of size $k$ such that $|X\cap V_i^j|=1$ for $i\in\{1,\ldots,k\}$. Let $s$ be the string with $p$ symbols that is the binary encoding of $j-1$. Consider the set $Y\subseteq \bigcup_{i=1}^p\{u_i,v_i\}$ such that for each $i\in\{1,\ldots,p\}$, $Y$ contains either $u_i$ or $v_i$, and $u_i$ is in $Y$ whenever $s[i]=1$. Observe that $Z=X\cup Y$ is a stable set of $G$ and it holds that 
$|Z\cap V_h|=1$ for each $h\in\{1,\ldots,p+k\}$. This means that  $(G,\{V_1,\ldots,V_{k+p}\})$  is a yes-instance of \probMIS.

For the opposite direction, assume that  $(G,\{V_1,\ldots,V_{k+p}\})$  is a yes-instance of \probMIS. Then there is a stable set  $Z$ of $G$ of size $k'=k+p$ such that 
$|Z\cap V_h|=1$ for each $h\in\{1,\ldots,p+k\}$. Let $Y=Z\cap\big(\bigcup_{i=1}^p\{u_i,v_i\}\big)$ and $X=Z\setminus Y$.
By the construction of color classes and because $Y$ is a stable set,  for each $i\in\{1,\ldots,p\}$, $Y$ contains either $u_i$ or $v_i$. Also, we have that $X\subseteq \bigcup_{j=1}^tV(G_j)$.
Consider the binary string $s$ of length $p$ such that 
$s[i]=1$ if $u_i\in Y$ and $s[i]=0$, otherwise, for all $i\in\{1,\ldots,p\}$. Notice that the vertices of $G_j$ such that $s$ is the binary encoding of $j-1$ are not adjacent to the vertices of $Y$ and for every $j'\in\{1,\ldots,t\}$ distinct from $j$, all the vertices of $G_{j'}$ are adjacent to at least one vertex of $Y$. This implies that $X\subseteq V(G_j)$. Therefore, $X$ is a stable set of $G_j$ of size $k$ and $|X\cap V_i^j|=1$ for $i\in\{1,\ldots,k\}$, that is, $(G_j,\{V_1^j,\ldots,V_k^j\})$ is a yes-instance of \probMIS. 

Notice that each vertex $v\in V(G_j)$ for $j\in\{1,\ldots,t\}$ is adjacent in $G$ to at most $n-1$ vertices of $G_j$ and $p$ vertices of $\bigcup_{i=1}^p\{u_i,v_i\}$. Therefore, the degeneracy of $G$ is at most $n+\log t$. Also, we have the number of color classes $k'=k+p\leq n+\log t$. Then because \probMIS is \classNP-complete and 
$(G,\{V_1,\ldots,V_{k+p}\})$ is a yes-instance of \probMIS if and only if there is $j\in\{1,\ldots,t\}$ such that $(G_j,\{V_1^j,\ldots,V_k^j\})$ is a yes-instance of \probMIS, the result of Bodlaender, Jansen, and Kratsch~\cite{BodlaenderJK14} implies that \probMIS does not admit a polynomial kernel unless $\classNP\subseteq\classCoNP/{\rm poly}$ when parameterized by the number of color classes $k$ and the degeneracy of the input graph. This concludes the proof.
\end{proof}

\section{\probStab on chordal graphs}\label{sec:perf}
For chordal graphs, we show that \probStab is \classFPT in the case of linear matroids when parameterized by $k$ by demonstrating a dynamic programming algorithm over tree decompositions exploiting representative sets~\cite{Lovasz77,Marx09,FominLPS16,LokshtanovMPS18}.

Let $\cM=(V,\cI)$ be a matroid and let $\mathcal{S}$ be a family of subsets of $V$. For a positive integer $q$, a subfamily $\widehat{\mathcal{S}}$ is \emph{$q$-representative for $\mathcal{S}$}  if the following holds: for every set $Y\subseteq V$ of size at most $q$, if there is a set $X\in\mathcal{S}$ disjoint from $Y$ with $X\cup Y\in \cI$ then there is $\widehat{X}\in\widehat{\mathcal{S}}$ disjoint from $Y$ with $\widehat{X}\cup Y\in \cI$. 
We write $\widehat{\mathcal{S}}\subseteq_{rep}^q\mathcal{S}$ to denote that $\widehat{\mathcal{S}}\subseteq\mathcal{S}$ is  $q$-representative for $S$. 
We use the results of of Fomin et al.~\cite{FominLPS16} to compute representative families for linear matroids. A family of sets $\mathcal{S}$ is said to be a \emph{$p$-family} for an integer $p\geq 0$ if $|S|=p$ for every $S\in\mathcal{S}$, and we use $\|A\|$ to denote the bit-length of the encoding of a matrix $A$.

\begin{proposition}[{\cite[Theorem 3.8]{FominLPS16}}]\label{prop:rep-rand}
Let  $M=(V,\mathcal{I})$ be a linear matroid and let $\mathcal{S}=\{S_1,\ldots,S_t\}$ be a $p$-family of independent sets. Then there exists    $\widehat{\mathcal{S}}\subseteq_{rep}^q\mathcal{S}$ of size at most $\binom{p+q}{p}$. Furthermore, given a representation $A$ of $M$ over a field $\mathbb{F}$, there is a randomized Monte Carlo algorithm computing  $\widehat{\mathcal{S}}\subseteq_{rep}^q\mathcal{S}$ of size at most $\binom{p+q}{p}$ in $\mathcal{O}(\binom{p+q}{p}tp^\omega+t\binom{p+q}{q}^{\omega-1})+\|A\|^{\mathcal{O}(1)}$ operations over $\mathbb{F}$, where 
$\omega$ is the exponent of matrix multiplication.\footnote{The currently best value is $\omega\approx 2.3728596$~\cite{AlmanW21}.}
\end{proposition}

The following theorem is proved by the bottom-up dynamic programming over a nice tree decomposition where representative sets are used to store partial solutions. 

\begin{theorem}\label{thm:FPT-chordal}
\probStab can be solved in $2^{\Oh(k)}\cdot \|A\|^{\Oh(1)}$ time by a one-sided error Monte Carlo algorithm with false negatives on frameworks with chordal graphs and linear matroids given by their representations $A$.
\end{theorem}

\begin{proof}
The algorithm uses a standard approach and, therefore, we only sketch the main ideas. Let $(G,\cM,k)$ be an instance of \probStab where $G$ is a chordal graph and $\cM$ is a linear matroid represented by a matrix $A$. 

We remind that a \emph{tree decomposition} of a graph~$G$
is a pair~$(T,\mathcal{X})$ where $T$ is a tree and $\mathcal{X} = \{X_i \mid i\in V(T)\}$ is a family of subsets of $V(G)$
such that
\begin{itemize}
\item $\bigcup_{t \in V(T)} X_t = V(G),$
\item for every edge~$e$ of~$G$ there is a $t\in V(T)$ such that $X_t$ contains both endpoints of~$e,$ and
\item for every~$v \in V(G),$ the subgraph of~${T}$ induced by $\{t \in V(T)\mid {v \in X_t}\}$ is connected.
\end{itemize}
The results of Gavril~\cite{Gavril74} imply that a graph $G$ is chordal if and only if $G$ admits a tree decomposition where each bag is a clique. Moreover, given a chordal graph $G$, a tree decomposition with clique bags (or, equivalently, a \emph{clique tree}) where $T$ has at most $n$ nodes can be constructed in linear time~\cite{RoseTL76,HabibMPV00}. 

A tree decomposition $\mathcal{T}=(T,\mathcal{X})$ of $G$ is  said to \emph{nice}  if $T$ is rooted in some node $r$ and
\begin{itemize}
\item $X_r=\emptyset$ and for any leaf node $l\in V(T)$, $X_l=\emptyset$,
\item every $t\in V(T)$ has at most two children,
\item if $t$ has one child $t'$ then
\begin{itemize}
\item either $X_{t} = X_{t'}\cup\{v\}$ for some $v\in V(G)\setminus X_{t'}$ and $t$ is called an {\em introduce node},
\item or $X_t = X_{t'}\setminus \{v\}$ for some $v\in X_{t'}$ and $t$  is called a {\em forget node},
\end{itemize}
\item if $t$ has two children $t_1$ and $t_2$ then $X_t= X_{t_1}= X_{t_2}$ and $t$ is called a {\em join node}.
\end{itemize}
By the results of Kloks~\cite{Kloks94}, we can turn in $\Oh(n^3)$ time a tree decomposition of a chordal graph into a nice tree decomposition where each bag is a clique and $T$ has at most $n^2$ nodes. 

Now we apply the bottom-up dynamic programming over a nice tree decomposition using the observation that a clique can contain at most one vertex of a stable set. For $t\in V(T)$, we denote by $T_t$ the subtree of $T$ rooted in $t$ and define $G_t=G[\bigcup_{t'\in V(T_t)}X_{t'}]$.
For every $t\in V(T)$, every subset $W\subseteq X_t$ of size at most one (that is, either $W=\{v\}$ for $v\in X_t$ or $W=\emptyset$), and every integer $p$ such that $|W|\leq p\leq k$, we compute a  $p$-family $R[t,W,p]$ of subsets of $V(G_t)$ that is $q=(k-p)$-representative for the family of all 
stable sets $S\subseteq V(G_t)$ of $G_t$ of size $p$ such that (i)~$S$ is independent with respect to $\cM$ and (ii) $S\cap X_t=W$. Notice that $(G,\cM,k)$ is a yes-instance of \probStab if and only if $R[r,\emptyset,k]\neq \emptyset$ and any set in $R[r,\emptyset,k]\neq \emptyset$ is a solution to the instance. For convenience, we assume that $R[t,W,0]=\emptyset$ if $|W|=1$. 
We use \Cref{prop:rep-rand} to ensure that $|R[t,W,p]|\leq \binom{k}{p}$.

If $t$ is a leaf node then $X_t=\emptyset$ and 
$R[t,\emptyset,p]=
\begin{cases}
\{\emptyset\}&\mbox{if } p=0,\\
\emptyset&\mbox{if } p\geq 1,
\end{cases}$ 
by the definition of $R[t,W,p]$.

Let $t$ be an introduce node with the child $t'$ and $X_t=X_{t'}\cup\{v\}$ for some $v\in V(G)\setminus X_{t'}$. 
For every $W\subseteq X_t$ of size at most one and every integer $p$ such that  $|W|\leq p\leq k$, we set 
\begin{equation*}
\mathcal{S}=\{S\cup\{v\}\colon S\in  R[t',\emptyset,p-1]\text{ and }S\cup \{v\}\in \cI\}
\end{equation*}
and 
use \Cref{prop:rep-rand} to compute $\widehat{\mathcal{S}}\subseteq_{rep}^q\mathcal{S}$ of size at most $\binom{p+q}{p}$ for $q=k-p$. Then we set  
\begin{equation*}
R[t,W,p]=
\begin{cases}
R[t',W,p]&\mbox{if }v\notin W,\\
\widehat{S}&\mbox{if }v\in W.
\end{cases}
\end{equation*}

Next, let $t$ be a forget node with the child $t'$ and $X_t=X_{t'}\setminus \{v\}$ for some $v\in X_{t'}$. For every $W\subseteq X_t$ of size at most one and every integer $p$ such that  $|W|\leq p\leq k$, we set 
\begin{equation*}
\mathcal{S}=R[t',\emptyset,p]\cup R[t',\{v\},p].
\end{equation*}
 We use \Cref{prop:rep-rand} to compute $\widehat{\mathcal{S}}\subseteq_{rep}^q\mathcal{S}$ of size at most $\binom{p+q}{p}$ for $q=k-p$. Then we set 
\begin{equation*}
R[t,W,p]=
\begin{cases}
\widehat{\mathcal{S}}&\mbox{if }W=\emptyset,\\
R[t',W,p]&\mbox{if }W\neq\emptyset.
\end{cases}
\end{equation*}

Finally, suppose that $t$ is a join node with  the children $t_1$ and $t_2$. 
For every $W\subseteq X_t$ of size at most one and every integer $p$ such that  $|W|\leq p\leq k$, we set 
\begin{equation*}
\mathcal{S}=\bigcup_{h=0}^p \{S\cup S'\colon S\in R[t_1,W,h],~S'\in R[t_2,W,p-h+|W|],\text{ and }S\cup S'\in\cI\}.
\end{equation*}
Note that $\mathcal{S}$ is a $p$-family. We use \Cref{prop:rep-rand} to compute $\widehat{\mathcal{S}}\subseteq_{rep}^q\mathcal{S}$ of size at most $\binom{p+q}{p}$ for $q=k-p$.
Then we set $R[t,W,p]=\widehat{\mathcal{S}}$.

The correctness of computing the families $R[t,W,p]$ follows from the description and the definition of representative sets. The arguments are completely standard for the bottom-up dynamic programming over tree decompositions and we leave the details to the reader.

To evaluate the running time, observe that for each $t\in V(T)$, every $W\subseteq X_t$ of size at most one, and every integer $p$ such that  $|W|\leq p\leq k$, we have that
$|R(t,W,p)|\leq\binom{k}{p}$. Because $|X_t|\leq n$ for each $t\in V(T)$ and $|W|\leq 1$, we obtain that for each $t$, we keep at most $2^k(n+1)$ families of sets of size at most $k$.  
Because $|V(T)|\leq n^2$, we have at most $2^kn^2(n+1)$ sets in total. 
Computing $R[t,W,p]$ for leaves takes a constant time.  
For introduce, forget, and join  nodes, we use \Cref{prop:rep-rand}. For an introduce node, we have that $|\mathcal{S}|\leq\binom{k}{p-1}$, for
a forget node, it holds that $|\mathcal{S}|\leq 2\binom{k}{p}$, and for a join node, $|\mathcal{S}|\leq \binom{k}{p}^2$ for each $t\in V(T)$, $W\subseteq X_t$, and $p$.  Thus, each $p$-family $\widehat{\mathcal{S}}$ is computed in $\Oh(\binom{k}{p}^3p^\omega+\binom{k}{p}^{\omega+1})+\|A\|^{\Oh(1)})$ time. This implies that $R[t,W,p]$ is computed in $\binom{k}{p}^{\Oh(1)}\cdot \|A\|^{\Oh(1)}$ time. Summarizing and observing that $n\leq \|A\|$, we obtain that the total running time is $2^{\Oh(k)}\cdot \|A\|^{\Oh(1)}$.  This concludes the proof.
\end{proof}

The algorithm in \Cref{thm:FPT-chordal} is randomized because it uses the algorithm from \Cref{prop:rep-rand} to compute representative sets. For some linear matroids, the algorithm can be derandomized using the deterministic construction of representative sets given by Lokshtanov et al.~\cite{LokshtanovMPS18}. In particular, this can be done for linear  matroids over any finite field and the field of rational numbers. 

We complement~\Cref{thm:FPT-chordal} by proving that it is unlikely that \probStab admits a polynomial kernel when parameterized by $k$ in the case of chordal graphs. 

\begin{theorem}\label{thm:nokern-chord}
\probStab on frameworks with chordal graphs and partition matroids does not admit a polynomial kernel when parameterized by $k$ unless $\classNP\subseteq\classCoNP/{\rm poly}$.
\end{theorem}

\begin{proof}
In the same way as in the proof of \Cref{thm:nokern-degen}, we prove that \probMIS  does not admit a polynomial kernel when parameterized by $k$ on chordal graphs unless $\classNP\subseteq\classCoNP/{\rm poly}$ where $k$ is the number of color classes. 

We construct a  cross-composition from \probMIS. 
Again, we say that two instances $(G,\{V_1,\ldots,V_k\})$ and $(G',\{V_1',\ldots,V_{k'}'\})$ are \emph{equivalent} if $|V(G)|=|V(G')|$ and $k=k'$. Consider $t$ equivalent instances 
$(G_i,\{V_1^i,\ldots,V_k^i\})$ of \probMIS for $i\in\{1,\ldots,t\}$ where each graph is chordal and has $n$ vertices. 
Then we construct the instance $(G,\{V_0,V_1,\ldots,V_{k}\})$ of \probMIS as follows.
\begin{itemize}
\item Construct the disjoint union of copies of $G_1,\ldots,G_t$.
\item Construct a clique $K$ with $t$ vertices $v_1,\ldots,v_t$.
\item For each $j\in\{1,\ldots,t\}$, make $v_j$ adjacent to all the vertices of every $G_i$ for $i\in \{1,\ldots,t\}$ that is distinct from $j$.
\item  Define $k+1$ color classes $V_0=K$ and  $V_i=\bigcup_{j=1}^tV_i^j$ for $i\in\{1,\dots,k\}$.
\end{itemize}
It is straightforward to see that $G$ is chordal  and 
the instance $(G,\{V_0,V_1,\ldots,V_{k}\})$ of \probMIS can be constructed in polynomial time.
We claim that  $(G,\{V_0,V_1,\ldots,V_{k}\})$ is a yes-instance of \probMIS if and only if there is $j\in\{1,\ldots,t\}$ such that $(G_j,\{V_1^j,\ldots,V_k^j\})$ is a yes-instance of \probMIS.

Suppose that $(G_j,\{V_1^j,\ldots,V_k^j\})$ is a yes-instance for some $j\in\{1,\ldots,t\}$. Then there is a stable set $X\subseteq V(G_j)$ of size $k$ such that $|X\cap V_i^j|=1$ for $i\in\{1,\ldots,k\}$. By the construction of $G$, the vertex $v_j\in K$ is not adjacent to any vertex of $G_j$. Thus, $Y=X\cup\{v_j\}$ is  stable set of $G$ such that $|Y\cap V_i|=1$ for each $i\in\{0,\ldots,k\}$. Therefore, $(G,\{V_0,V_1,\ldots,V_{k}\})$  is a yes-instance of \probMIS.

For the opposite direction, assume that  $(G,\{V_0,V_1,\ldots,V_{k}\})$  is a yes-instance of \probMIS. Then there is a stable set  $Y$ of $G$ of size $k+1$ such that 
$|Y\cap V_i|=1$ for each $i\in\{0,\ldots,k\}$. In particular, $|Y\cap V_0|=1$. Then there is $j\in\{1,\ldots,t\}$ such that $v_j\in Y$. By the construction of $G$, we have that $X=Y\setminus\{v_j\}\subseteq V(G_j)$.  Then $|X\cap V_i^j|=1$ for each $i\in\{1,\ldots,k\}$, that is, $(G_j,\{V_1^j,\ldots,V_k^j\})$ is a yes-instance of \probMIS. 

Le and Pfender in \cite{LeP14}  proved that \textsc{Rainbow Matching} remains \classNP-complete  on paths.  This implies that \probMIS is also \classNP-complete  on paths, and hence on chordal graphs. Because the number of color classes is $k+1\leq n+1$ and \probMIS is \classNP-complete on chordal graphs, we can apply the result of Bodlaender, Jansen, and Kratsch~\cite{BodlaenderJK14}. This concludes the proof.
\end{proof}

\section{Conclusion}\label{sec:conclusion}
In this paper, we investigated the parameterized complexity of the \probStab problem for various classes of graphs where the classical \probIS problem is tractable. We derived kernelization results and \classFPT algorithms, complemented by complexity lower bounds. We believe exploring \probStab on other natural graph classes with similar properties would be interesting. For instance,  \probIS is solvable in polynomial time on claw-free graphs~\cite{Minty80} and AT-free graphs~\cite{BroersmaKKM99}. While our unconditional lower bound from \Cref{thm:lb-uncond} applies to these classes, it does not rule out the possibility of \classFPT algorithms for frameworks with \emph{linear} matroids. A similar question arises regarding graphs with a polynomial number of minimal separators~\cite{BouchitteT01,BouchitteT02}.

\bibliographystyle{siam}

\bibliography{Frameworks}

\end{document}